\documentclass[a4paper]{article}
	\usepackage[pdftex]{graphicx}
	\usepackage{geometry}
	\usepackage{epstopdf}
	\usepackage[english]{babel}
	\usepackage{color}
	\usepackage{amssymb}
	\usepackage{amsmath}
	\usepackage{amsthm}
	\usepackage{bbm}
	\usepackage{tikz}
	\usepackage{multirow}
	\usepackage{float}
	\usepackage[ruled]{algorithm2e}
	\newtheorem{theorem}{Theorem}
	
	\newtheorem{corollary}{Corollary}

	\newtheorem{example}{Example}

\title{A 12/7-approximation algorithm for the discrete Bamboo Garden Trimming problem}
\author{Martijn van Ee \thanks{Netherlands Defence Academy, Faculty of Military Sciences, Den Helder, The Netherlands, M.v.Ee.01@mindef.nl }}
	
\begin{document}
\maketitle

\begin{abstract}
We study the discrete Bamboo Garden Trimming problem (BGT), where we are given $n$ bamboos with different growth rates. At the end of each day, one can cut down one bamboo to height zero. The goal in BGT is to make a perpetual schedule of cuts such that the height of the tallest bamboo ever is minimized. Here, we improve the current best approximation guarantee by designing a $12/7$-approximation algorithm. This result is based on a reduction to the Pinwheel Scheduling problem. We show that a guarantee of $12/7$ is essentially the best we can hope for if our algorithm is based on this type of reduction.\\
\textbf{Keywords:} bamboo garden trimming, pinwheel scheduling, approximation algorithms
\end{abstract}

\section{Introduction}
In the discrete Bamboo Garden Trimming problem (BGT), we are given a garden with $n$ bamboos, where bamboo $i$ grows with a rate of $h_i$ per day. Initially, all bamboos have height zero. At the end of each day, one can cut down one bamboo to height zero. The goal in BGT is to make a perpetual schedule of cuts (trimming schedule) such that the height of the tallest bamboo ever is minimized. The problem was introduced by Gasieniec et al. \cite{GKLLMR17}. They also introduced a continuous version of BGT, in which one can cut a bamboo at any time. However, here the bamboos are located at vertices in a weighted graph, and time passes while traveling. The discrete and continuous BGT have, among others, applications in scheduling maintenance of machines. Gasieniec et al. gave a 2-approximation for discrete BGT. Recently, Della Croce \cite{DellaCroce2020} improved this to a guarantee of approximately 1.888. He also showed that the approximation guarantee of his algorithm converges to $12/7$ when $h_1 \gg \sum_i h_i$. In this paper, we will extend the approach of \cite{GKLLMR17} in order to obtain a $12/7$-approximation algorithm. This result is obtained by reducing BGT to the Pinwheel Scheduling problem (PS).

In the Pinwheel Scheduling problem \cite{Holteetal1989}, we are given $n$ jobs with integral periods $p_1\leq\ldots\leq p_n$. On each day, we can schedule at most one job. The goal is to design a perpetual schedule such that each job $i$ is scheduled at least once in any period of $p_i$ consecutive days, or to conclude that no such schedule exists. Holte et al. \cite{Holteetal1989} showed that the problem is contained in PSPACE, but no completeness result is known. The strongest hardness result is given by Jacobs and Longo \cite{JacobsLongo2014}, who showed that there is no (pseudo-)polynomial time algorithm for the Pinwheel Scheduling problem, unless Satisfiability can be solved by a randomized algorithm running in expected time $n^{O(\log{n}\log\log{n})}$. An important issue of PS, and other related problems like BGT, is the representation of a solution. Explicitly writing down the schedule may take exponential time and space. In this paper, we only consider periodic schedules that can be represented by numbers $o_i$ and $t_i$ for each $i$. In a periodic schedule, job $i$ is scheduled for the first time on day $o_i$, and then scheduled every $t_i$ days later. Our algorithm will be able to compute the $o_i$'s and $t_i$'s in polynomial time (Section \ref{sec:poly}).

It is quite easy to reduce PS to BGT. Given an instance of PS, we create an instance of BGT by creating a bamboo $i$ for each job $i$ with growth rate $h_i=1/p_i$. Now, there is a feasible schedule for PS if and only if there is a trimming schedule for BGT with maximum height 1. Hence, BGT is harder than PS, and the hardness result by \cite{JacobsLongo2014} also holds for BGT.

There are several problems that relate to the Pinwheel Scheduling problem. The continuous version of PS, in which jobs can be scheduled at any time, but time passes when one is traveling, is known as the Periodic Latency problem \cite{CSW11}. Moreover, the problem also arises as a subproblem is some variants of inventory routing \cite{Balleretal2019} and replenishment problems \cite{Bosmanetal2018}.

An interesting aspect of a PS-instance is the density. The density of instance $A=\{p_1,\ldots,p_n\}$ is defined as $\rho(A)=\sum_i 1/p_i$. Clearly, if $\rho(A)>1$, instance $A$ is not schedulable. Holte et al. \cite{Holteetal1989} showed that if $p_i$ divides $p_j$ whenever $i<j$, and $\rho(A)\leq 1$, then $A$ is schedulable. Consequently, if $\rho(A)\leq 1/2$, $A$ can be scheduled by rounding each period down to its nearest power of 2. Later, more elaborate ways of rounding the periods were considered \cite{ChanChin1992,ChanChin1993,FishburnLagarias2002}. In the first paper by Chan and Chin \cite{ChanChin1992}, they designed an algorithm that was able to round any instance $A$ with $\rho(A)\leq 2/3$ to a schedulable instance. They also considered a simpler algorithm such that any instance with density $7/12$ is rounded to a schedulable instance. We will use this algorithm for our main result. They also observed that instance $A=\{2,3,M\}$ cannot be scheduled for any value of $M$. Hence, there exists an instance with density above $5/6$ that is not schedulable. They conjectured that any instance $A$ with $\rho(A)\leq 5/6$ is schedulable. In the second paper by Chan and Chin \cite{ChanChin1993}, they improved the achievable density bound to $7/10$. Finally, Fishburn and Lagarias \cite{FishburnLagarias2002} showed that any instance with density at most $3/4$ can be scheduled.

In the next section, we will describe the 2-approximation algorithm designed by \cite{GKLLMR17}. In Section \ref{sec:main} we will show how we can use the algorithm by Chan and Chin \cite{ChanChin1992} to improve the approximation guarantee to $12/7$. This sections will focus on the approximation guarantee. In Section \ref{sec:poly}, we will discuss the running time of the algorithms.

\section{A $2$-approximation algorithm}
In this section, we will describe the approach taken by \cite{GKLLMR17} in order to obtain a 2-approximation algorithm. For this, define $H=\sum_i h_i$. As observed in \cite{GKLLMR17}, $H$ is a lower bound on the optimal value. Hence, if we find a solution with value $\alpha H$, we know that we are within a factor $\alpha$ of the optimal value.
Given an instance of BGT, i.e., $h_1,\ldots,h_n$, we define the following periods:
\[ p_i := \frac{2H}{h_i}. \]
The period $p_i$ can be interpreted as the (possibly fractional) number of days it takes bamboo $i$ to reach height $2H$. Since the $p_i$'s need not to be integral, we call $A=\{p_1,\ldots,p_n\}$ a pseudo-instance of the Pinwheel Scheduling problem. A feasible solution for a pseudo-instance of the Pinwheel Scheduling problem is a schedule assigning at most one job per day such that each job $i$ is scheduled at least once in any period of $\lfloor p_i \rfloor$ consecutive days. If we are able to find a feasible schedule for the created pseudo-instance, we know that no bamboo will ever exceed height $2H$. Hence, we will obtain a 2-approximation algorithm.

The density of the created pseudo-instance $A$ is equal to 
\[ \rho(A)=\sum_i \frac{1}{p_i} = \sum_i \frac{h_i}{2H} = \frac{1}{2}. \]
If we round every period down to its nearest power of 2, we get an instance of the Pinwheel Scheduling problem with density at most 1. So, the result from \cite{Holteetal1989} ensures there is a schedule for this instance of the Pinwheel Scheduling problem. In Section \ref{sec:poly}, it is shown how to obtain a representation of a feasible schedule in polynomial time. Since periods are rounded down, the resulting schedule can also be used as a feasible solution for the pseudo-instance $A$. This concludes the proof.

\section{A $12/7$-approximation algorithm}
\label{sec:main}

Our algorithm uses the same approach as the 2-approximation by \cite{GKLLMR17}. Given an instance of BGT, we create a pseudo-instance of the Pinwheel Scheduling problem by defining periods:
\[ p_i := \frac{12H}{7h_i}. \]
The pseudo-instance $A=\{p_1,\ldots,p_n\}$ has density $\rho(A)=7/12$. As noted in the introduction, Chan and Chin \cite{ChanChin1992} designed an algorithm that is able to round any instance of PS with density at most $7/12$ to a schedulable instance. As it turns out, this algorithm is even able to round any pseudo-instance of PS with density at most $7/12$ to a schedulable instance of PS (Theorem \ref{thm:cc}). In Section \ref{sec:poly}, we show how to obtain a representation of a feasible schedule in polynomial time. Hence, we obtain a $12/7$-approximation algorithm for BGT.

To explain the algorithm by \cite{ChanChin1992}, we need to define the operation specialization. If we specialize a pseudo-instance of PS with respect to $\{x\}$, it means that we round every period down to its nearest integer in $\{x,2x,4x,\ldots,x2^j,\ldots\}$. For example, in the algorithm by \cite{GKLLMR17}, the pseudo-instance is specialized with respect to $\{2\}$. We can also specialize a pseudo-instance with respect to multiple integers. If we specialize a pseudo-instance of PS with respect to $\{x,y\}$, it means that we round every period down to its nearest integer in 
\[ \{x,2x,4x,\ldots,x2^j,\ldots\} \cup \{y,2y,4y,\ldots,y2^j,\ldots\}. \]
The algorithm by Chan and Chin \cite{ChanChin1992} uses specialization with respect to $\{2,3\}$. We use the following notation. Let $A = A_2 \cup A_3$ be a pseudo-instance, where $A_2 = \{p_i \in A| 2(2^j) \leq p_i < 3(2^j), \text{for some integer }j \geq 0\}$ and $A_3 = \{p_i \in A|3(2^j) \leq p_i < 2(2^{j+1}), \text{for some integer } j \geq 0\}$. Let $B$ and $C$ be the specialization of $A_2$ and $A_3$ with respect to $\{2\}$ and $\{3\}$, respectively. Then, if $\lceil 2\rho(B)\rceil/2 + \lceil 3\rho(C)\rceil/3 \leq 1$, the instance is schedulable (Theorem 3.1.1 in \cite{ChanChin1992}). 

After specializing $A$ with respect to $\{2,3\}$, we may want to further adjust $B$ and $C$. For this, define $P\subseteq B$ and $Q \subseteq C$ such that
\begin{align*}
\rho(B) = \frac{r}{2} + \rho(P) \text{ with } 0 < \rho(P) < \frac{1}{2} \text{ for some integer } r \geq 0, \\
\rho(C) = \frac{s}{3} + \rho(Q) \text{ with } 0 < \rho(Q) < \frac{1}{3} \text{ for some integer } s \geq 0.
\end{align*}
If $\rho(P)+\rho(Q)=0$, no further action is taken. If $\rho(P)+\rho(Q)>0$, we will start our normalization procedure, which considers four cases. Here, treating set $X$ as one $z$ means that we will reserve space in the schedule as if the jobs in set $X$ were one $z$.
\begin{itemize}
	\item[(a)] Treat $P$ and $Q$ as one $3$. Here, we specialize $P$ with respect to $\{3\}$, and then move it to $C$. This will be done if $4\rho(P)/3+\rho(Q) \leq 1/3$.
 	\item[(b)] Treat $P$ and $Q$ as one $2$. Here, we specialize $Q$ with respect to $\{2\}$, and then move it to $B$. This will be done if $1/3 < 4\rho(P)/3+\rho(Q) \leq 2/3$ and $\rho(P)+3\rho(Q)/2 \leq 1/2$.
	\item[(c)] Treat $P$ and $Q$ as two $3$'s. Here, we specialize $P$ with respect to $\{3\}$, and then move it to $C$. This will be done if $1/3 < 4\rho(P)/3+\rho(Q) \leq 2/3$ and $\rho(P)+3\rho(Q)/2 > 1/2$.
	\item[(d)] Treat $P$ and $Q$ as a $2$ and a $3$. This will be done if $4\rho(P)/3+\rho(Q) > 2/3$.
\end{itemize}
We define $B'$ and $C'$ to be $B$ and $C$ after normalization, respectively. The following theorem states that if we specialize a pseudo-instance $A$ of PS with respect to $\{2,3\}$, and then apply normalization, we obtain a schedulable instance whenever $\rho(A)\leq 7/12$. The proof is actually the same as the one in \cite{ChanChin1992}, and is therefore given in the appendix.

\begin{theorem}
\label{thm:cc}
The algorithm of Chan and Chin \cite{ChanChin1992} obtains a schedulable PS-instance for any pseudo-instance $A$ with $\rho(A)\leq 7/12$. 
\end{theorem}

\begin{corollary}
There is a $12/7$-approximation algorithm for the discrete Bamboo Garden Trimming problem.
\end{corollary}

Next, we will illustrate that the approximation guarantee of $12/7$ is essentially the best we can hope for when our algorithm is based on a reduction to a pseudo-instance of PS. First, we give an example of a pseudo-instance of PS with density $7/12+\delta$, with $\delta>0$, that is not schedulable.

\begin{example}
\label{ex1}
Consider the pseudo-instance with $p_1=3-\epsilon$, $p_2=4-\epsilon$, and $p_3=M$, with $\epsilon>0$ and $M$ a large number. This pseudo-instance has density
\[ \frac{7}{12} + \delta \text{ with } \delta = \frac{\epsilon}{9-3\epsilon} + \frac{\epsilon}{16-4\epsilon} + \frac{1}{M}. \]
In order to obtain a feasible schedule for this pseudo-instance, we should round $p_1$ to 2. Similarly, we should round $p_2$ to an integer smaller than or equal to 3. In the most conservative rounding, we obtain the PS-instance $\{2,3,\lfloor M \rfloor\}$, which is known to be non-schedulable \cite{ChanChin1992}. Hence, our pseudo-instance is not schedulable.
\end{example}

The example above shows that if we reduce instances of BGT to a pseudo-instances of PS, and then only focus on schedulability of these pseudo-instances, we will not improve upon the guarantee of $12/7$. However, it might be possible that ``bad'' pseudo-instances like the one in Example \ref{ex1} are not encountered after the reduction. The next example shows that if we use $H$ as our lower bound, and a magnifying factor of $12/7-\eta$, to make the reduction from BGT to a pseudo-instance of PS, we will encounter the pseudo-instance of Example \ref{ex1}. This shows that if we use $H$ as our lower bound, algorithms based on a reduction to a pseudo-instance of PS will not give guarantees better than $12/7$.

\begin{example}
Consider the BGT-instance with $h_1=4$, $h_2=3$, and $h_3=\gamma$, with $0<\gamma<49\eta/(12-7\eta)$. Since $H=7+\gamma$, we obtain the pseudo-instance of PS with
\begin{align*}
p_1 &= \left( \frac{12}{7} - \eta \right)\left( \frac{7+\gamma}{4} \right) = 3 + \frac{3\gamma}{7} - \frac{(7+\gamma)\eta}{4} = 3 - \epsilon_1 \\
p_2 &= \left( \frac{12}{7} - \eta \right)\left( \frac{7+\gamma}{3} \right) = 4 + \frac{4\gamma}{7} - \frac{(7+\gamma)\eta}{3} = 4 - \epsilon_2 \\
p_3 &= \left( \frac{12}{7} - \eta \right)\left( \frac{7+\gamma}{\gamma} \right) = \frac{12}{\gamma} + \frac{12}{7} - \frac{(7+\gamma)\eta}{\gamma} = M
\end{align*}
\end{example}

We like to note that the chosen approach might be able to improve upon the current guarantee, if another lower bound is used to reduce to PS. For example, Della Croce \cite{DellaCroce2020} observed that, if $n\neq 1$, the optimal value is at least $2h_1$. Hence, we could use $\max\{2h_1,\sum_i h_i\}$ as a lower bound.

\section{Running time}
\label{sec:poly}
First observe that rounding the periods as is done by the algorithm from \cite{ChanChin1992} can be done in polynomial time. Here, we give an algorithm that, for instances obtained by applying the algorithm from \cite{ChanChin1992}, produces a representation of a feasible solution in polynomial time. For this, it is sufficient to show that such an algorithm exists for instances $A$ with the property that $p_i$ divides $p_j$ whenever $p_i\leq p_j$, and $\rho(A)\leq 1$. In \cite{Holteetal1989}, an algorithm for instances with this property is given. However, this algorithm runs in exponential time, since it constructs a schedule of length $2\prod_i p_i$. For instances with the aforementioned property, we can restrict ourselves to periodic schedules \cite{Holteetal1989}. As said in the introduction, these can be represented by an offset time $o_i$ and a time $t_i$ between two consecutive days on which job $i$ is scheduled, for each $i$. We will show that we can find these values in polynomial time. 

\begin{theorem}
If an instance $A$ satisfies the property that $p_i$ divides $p_j$ whenever $p_i\leq p_j$, and $\rho(A)\leq 1$, then a representation of a feasible schedule can be found in polynomial time.
\end{theorem}
\begin{proof}
First, we show that a feasible schedule exists. Then, we show that this proof leads to a polynomial time algorithm. For proving existence of a feasible schedule, we use induction on the size of the instance. By our assumption on the periods, we can partition an instance $A$ into $k\leq p_1$ subsets $S_1,\ldots,S_k$ with $\sum_{i\in S_j} 1/p_i \leq 1/p_1$ for all $j$. Moreover, $p_i/p_1\in\mathbb{N}$ for all $i$.

Create a new instance $A'_j$ containing the jobs from $S_j$ and set the periods to $p'_i = p_i/p_1$ for $i \in S_j$. Clearly, $\sum_i 1/p'_i\leq 1$ holds for this instance. Suppose there is a feasible schedule for $A'_j$, and let $\sigma_j(i,\ell)$ denote the day on which job $i$ is scheduled for the $\ell$th time. Now, in the original instance, schedule job $i\in S_j$ for the $\ell$th time (denoted by $\sigma(i,\ell)$) on day:
\begin{equation}
\label{eq1}
\sigma(i,\ell) = j + (\sigma_j(i,\ell)-1)\cdot p_1.
\end{equation}

To see that this schedule is feasible, note that clients from set $S_j$ are only scheduled on days $\tau$ for which $\tau \equiv j \mod p_1$, and therefore there are no two clients scheduled on the same day. Furthermore, the time between two consecutive days on which a job is scheduled is exactly $p_1$ times as long as in the schedule for $B'_j$. Therefore instance $A$ will have a feasible schedule if the subinstances $A'_j$ are schedulable. As the inductive hypothesis trivially holds for instances with one client, the result follows.

The approach above can be used to find $t_i$ and $o_i$ in polynomial time. Clearly, $t_i$ is equal to $p_i$. To obtain the values of $o_i$ for all $i$, we repeatedly partition the instance as is done above, until only singletons are left. After partitioning, each subinstance is scaled down by a factor $p_1$. Now, the value of $o_i$ can be found by using equation \ref{eq1} recursively. Since in each step, we create at least one subset that contains only one job, this procedure will run in polynomial time.
\end{proof}

Finally, we will show that the algorithm above can be used to create a representation of a feasible schedule for instances produced by the algorithm of \cite{ChanChin1992}. In each case, the values of $t_i$ are equal to the periods found by the algorithm. In case our instance is only specialized with respect to either $\{2\}$ or $\{3\}$, we can directly apply the algorithm above. In case our instance is specialized with respect to $\{2,3\}$, we will assign the jobs in $B'$ to the odd days, and the jobs in $C'$ to the even days. If $B'$ consists of a single $2$, then this job is scheduled on every odd day. Otherwise, the periods are divided by 2, the algorithm above is applied, and the resulting schedule is stretched by a factor 2. A similar algorithm works for scheduling the jobs in $C'$.

\section{Conclusion}
In this paper, we presented a $12/7$-approximation algorithm for the discrete Bamboo Garden Trimming problem. This improves upon the $1.888$-approximation by Della Croce \cite{DellaCroce2020}. The algorithm uses a reduction to the Pinwheel Scheduling problem, similar to the algorithm designed by Gasieniec et al. \cite{GKLLMR17}. We also gave examples which show that the current bound is essentially the best we can hope for if we restrict ourselves to this approach.

For future research, it would be interesting to see if other approaches can improve upon the guarantee of $12/7$. It might also be interesting to study lower bounds on approximability. Moreover, it would be great if one is able to prove that either the discrete Bamboo Garden Trimming problem or the Pinwheel Scheduling problem is complete for some class, e.g., PSPACE-complete. 

\paragraph{Acknowledgements} We would like to thank Thomas Bosman for useful discussions on the computational efficiency of algorithms for the Pinwheel Scheduling problem.

\bibliography{BibBamboo}
\bibliographystyle{plain}

\appendix
\section{Proof of Theorem \ref{thm:cc}}

Using the definitions from Section \ref{sec:main}, we have
\begin{align*}
\rho(A_2) > \frac{2\rho(B)}{3} = \frac{r}{3} + \frac{2\rho(P)}{3}, \\
\rho(A_3) > \frac{3\rho(C)}{4} = \frac{s}{4} + \frac{3\rho(P)}{4}.
\end{align*}
Hence, we get
\[ \frac{r}{3} + \frac{s}{4} + \frac{2\rho(P)}{3} + \frac{3\rho(Q)}{4} < \rho(A_2) + \rho(A_3) = \rho(A) \leq \frac{7}{12}. \]
Hence, we only have combinations of $(r,s)$ in
\[ \{(0,0), (0,1), (0,2), (1,0), (1,1)\}. \]

We will now check that in each of the cases that the algorithm considers, we have $y = \lceil 2\rho(B')\rceil/2 + \lceil 3\rho(C')\rceil/3 \leq 1$.

If $\rho(P) + \rho(Q) = 0$, we have $\rho(B') = r/2$, $\rho(C') = s/3$, and $y = r/2 + s/3$. As $y < 1$ for all possible cases of $r$ and $s$, the instance can be scheduled.
If $\rho(P) + \rho(Q) > 0$, we consider the following cases.

\textit{Case (a)} : treat $P$ and $Q$ as a single $3$. After normalization, $\rho(B') = r/2$ and $s/3 < \rho(C') \leq (s + 1)/3$. Thus, $y = r/2 + (s + 1)/3$. Since $r/3 + s/4 < 7/12$, we only have to consider the following cases
\[ (r, s) \in \{(0, 0), (0, 1), (0, 2), (1, 0)\}, \]
and, in all cases, $y < 1$.

\textit{Case (b)}: treat $P$ and $Q$ as a single $2$. After normalization, $r/2 < \rho(B')\leq(r + 1)/2$ and $\rho(C') = s/3$. Thus, $y = (r + 1)/2 + s/3$. From the conditions for case (b), $4\rho(P)/3 + \rho(Q) > \frac{1}{3}$. With $\rho(P) < 1/2$, we have
\[ \frac{2\rho(P)}{3} + \frac{3\rho(Q)}{4} = \frac{3}{4}\left(\frac{4\rho(P)}{3} + \rho(Q)\right) - \frac{\rho(P)}{3} > \frac{3}{4}\cdot\frac{1}{3} - \frac{1}{3}\cdot\frac{1}{2} = \frac{1}{12}. \]
Thus, 
\begin{align*}
\frac{r}{3} + \frac{s}{4} + \frac{1}{12} < \frac{7}{12} \\ 
\frac{r}{3} + \frac{s}{4} < \frac{1}{2}. 
\end{align*}
The only possible cases for $(r, s)$ are $(0, 0)$, $(0, 1)$, and $(1, 0)$, and this gives $y \leq 1$.

\textit{Case (c)}: treat $P$ and $Q$ as two $3$'s. After normalization, $\rho(B')= r/2$ and $(s + 1)/3 < \rho(C') \leq (s + 2)/3$. Thus, $y = r/2 + (s + 2)/3$. From the conditions for this case, $\rho(P) + 3\rho(Q)/2 > 1/2$. With $\rho(Q) < 1/3$, we have
\[ \frac{2\rho(P)}{3} + \frac{3\rho(Q)}{4} = \frac{2}{3}\left(\rho(P) + \frac{3\rho(Q)}{2}\right) - \frac{\rho(Q)}{4} > \frac{2}{3}\cdot\frac{1}{2} - \frac{1}{4}\cdot\frac{1}{3} = \frac{1}{4}. \]
Thus, 
\begin{align*}
\frac{r}{3} + \frac{s}{4} + \frac{1}{4} < \frac{7}{12} \\ 
\frac{r}{3} + \frac{s}{4} < \frac{1}{3}. 
\end{align*}
The only possible cases for $(r, s)$ are $(0, 0)$ and $(0, 1)$, and this gives $y < 1$.

\textit{Case (d)} : treat $P$ and $Q$ as a $2$ and a $3$. After normalization, $r/2 < \rho(B') \leq (r + 1)/2$ and $s/3 < \rho(C') \leq (s + 1)/3$. Thus, $y = (r + 1)/2 + (s + 1)/3$. From the conditions for this case, $4\rho(P)/3 + \rho(Q) > 2/3$. With $\rho(P) < 1/2$, we have
\[ \frac{2\rho(P)}{3} + \frac{3\rho(Q)}{4} = \frac{3}{4}\left(\frac{4\rho(P)}{3} + \rho(Q)\right) - \frac{\rho(P)}{3} > \frac{3}{4}\cdot\frac{2}{3} - \frac{1}{3}\cdot\frac{1}{2} = \frac{1}{3}. \]
Thus,
\begin{align*}
\frac{r}{3} + \frac{s}{4} + \frac{1}{3} < \frac{7}{12} \\ 
\frac{r}{3} + \frac{s}{4} < \frac{1}{4}. 
\end{align*}
The only possible case for $(r, s)$ is $(0, 0)$ and this gives $y < 1$. \qed

\end{document}